\documentclass[conference]{IEEEtran}
\IEEEoverridecommandlockouts

\usepackage[english]{babel}

\usepackage[firstpage]{draftwatermark}
\SetWatermarkText{\hspace*{8in}\raisebox{9.2in}
  {\includegraphics[scale=0.1]{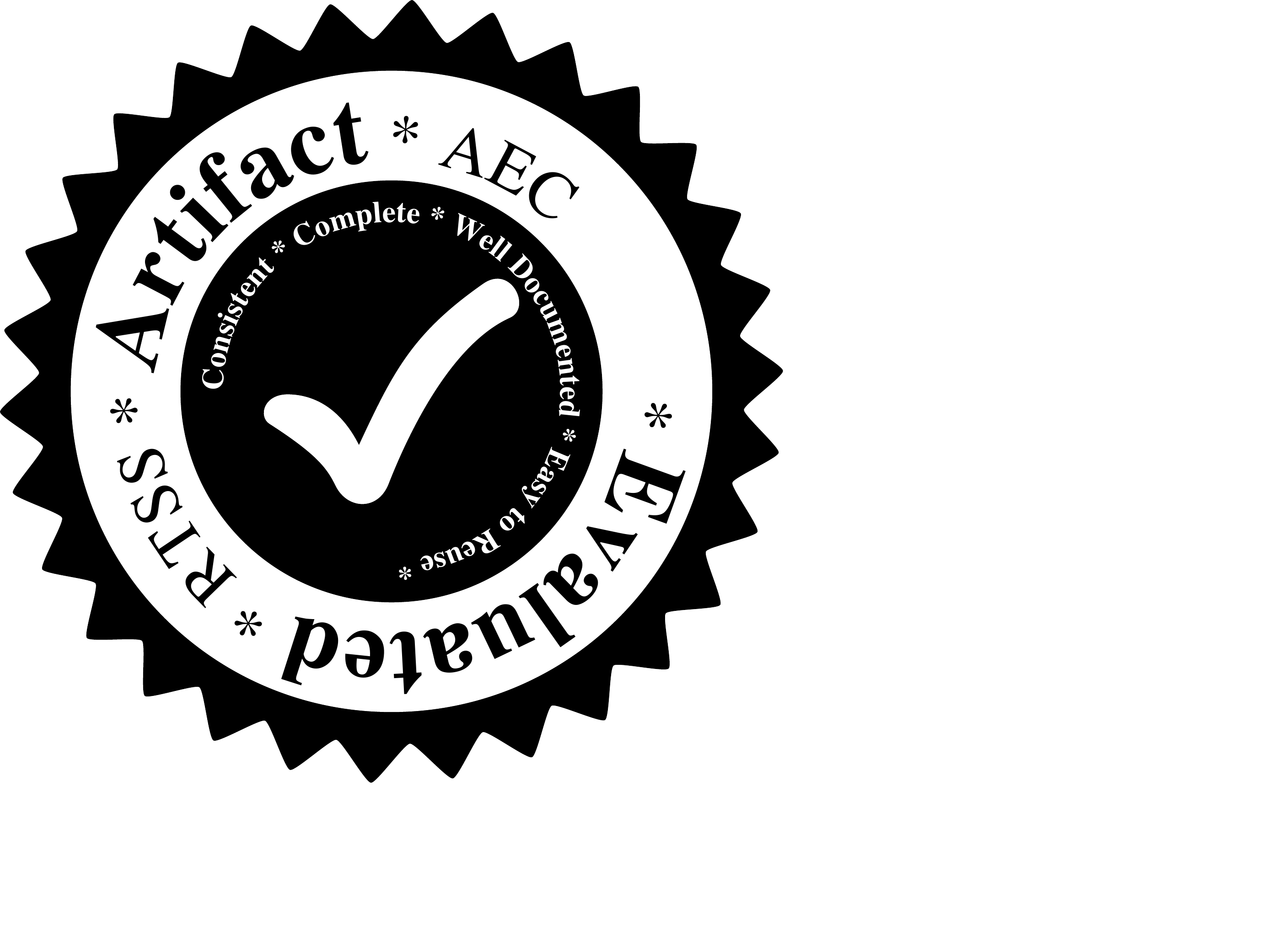}}}
\SetWatermarkAngle{0}

\usepackage{amsmath,amssymb,amsfonts,amsthm}
\usepackage{thmtools}
\usepackage{mathtools}

\DeclarePairedDelimiter\floor{\lfloor}{\rfloor}
\usepackage{algorithm}
\usepackage{algpseudocode}
\usepackage{graphicx}
\usepackage{textcomp}
\usepackage{siunitx}
\usepackage{slashbox}
\usepackage{xcolor}
\usepackage{multicol}
\usepackage{hyperref}
\usepackage{listings}
\usepackage{subcaption}
\usepackage{cite}

\algnewcommand{\algorithmicand}{\textbf{ and }}
\algnewcommand{\algorithmicor}{\textbf{ or }}
\algnewcommand{\OR}{\algorithmicor}
\algnewcommand{\AND}{\algorithmicand}
\algnewcommand{\var}{\texttt}

\newtheorem{assumptions}{Assumption}
\newtheorem{thm}{Theorem}
\newtheorem{lemma}{Lemma}

\begin{document}
\title{Cluster-based Network Time Synchronization \\
  for Resilience with Energy Efficiency\\
}




\author{
	\IEEEauthorblockN{{Nitin Shivaraman$^{1}$}, {Patrick Schuster$^{2}$}, {Saravanan Ramanathan$^{1}$}, {Arvind Easwaran$^{3}$}, {Sebastian Steinhorst$^{2}$}}\\	
	\IEEEauthorblockA{$^{1}$TUMCREATE, Singapore,
		$^{2}$Technical University of Munich, Germany,
		$^{3}$Nanyang Technological University, Singapore}
    \normalsize {$^{1}$\{nitin.shivaraman, saravanan.ramanathan\}@tum-create.edu.sg, $^{2}$\{patrick.schuster, sebastian.steinhorst\}@tum.de, $^{3}$arvinde@ntu.edu.sg}}

%
%
%
%
%
%
\maketitle

\begin{abstract}
Time synchronization of devices in Internet-of-Things (IoT) networks is one of the challenging problems and a pre-requisite for the design of low-latency applications. Although many existing solutions have tried to address this problem, almost all solutions assume all the devices (nodes) in the network are faultless. Furthermore, these solutions exchange a large number of messages to achieve synchronization, leading to significant communication and energy overhead. To address these shortcomings, we propose \emph{C-sync}, a clustering-based decentralized time synchronization protocol that provides resilience against several types of faults with energy-efficient communication. C-sync achieves scalability by introducing multiple reference nodes in the network that restrict the maximum number of hops any node can have to its time source. The protocol is designed with a modular structure on the Contiki platform to allow application transitions. We evaluate C-sync on a real testbed that comprises over 40 Tmote Sky hardware nodes distributed across different levels in a building and show through experiments the fault resilience, energy efficiency, and scalability of the protocol. C-sync detects and isolates faults to a cluster and recovers quickly. The evaluation makes a qualitative comparison with state-of-the-art protocols and a quantitative comparison with a class of decentralized protocols (derived from GTSP) that provide synchronization with no/limited fault-tolerance. Results also show a reduction of 56.12\% and 75.75\% in power consumption in the worst-case and best-case scenarios, respectively, compared to GTSP, while achieving similar accuracy.
\end{abstract}

%
%

\section{Introduction}

Most Internet-of-Things (IoT) networks comprise resource-constrained, battery-operated nodes that are interconnected through a wired or wireless medium. 
Communication among these nodes forms the bulk of their operation to exchange information~\cite{Ganeriwal:2003:TPS:958491.958508}.
Real-time studies start with a time-synchronized network as the foundation for communication among the nodes. 
With limited communication during time synchronization, the information exchange among the nodes must be trustworthy.
Hence, a synchronization solution with fault resilience is necessary, in addition to maintaining accuracy and efficiency.
  
A time synchronization protocol needs to ensure stability in synchronization throughout the network. 
Wireless sensor networks are often plagued by error-prone nodes that can result in faults including, but not limited to, selective forwarding and tampered data (spikes, outliers, etc.)~\cite{fault_survey}.
These faults form a sub-class of byzantine faults observed in radio communication and could result in major deadline misses and power dissipation due to erroneous time information, leading to destabilization of the network. 
A faulty node with incorrect information can jeopardize the entire network if not addressed~\cite{7397831}.
Most importantly, synchronization protocols must be resilient to such faults and ensure functional correctness to minimize potential network downtimes.

 \begin{figure}
	\centering
	\includegraphics[width = 0.8\columnwidth]{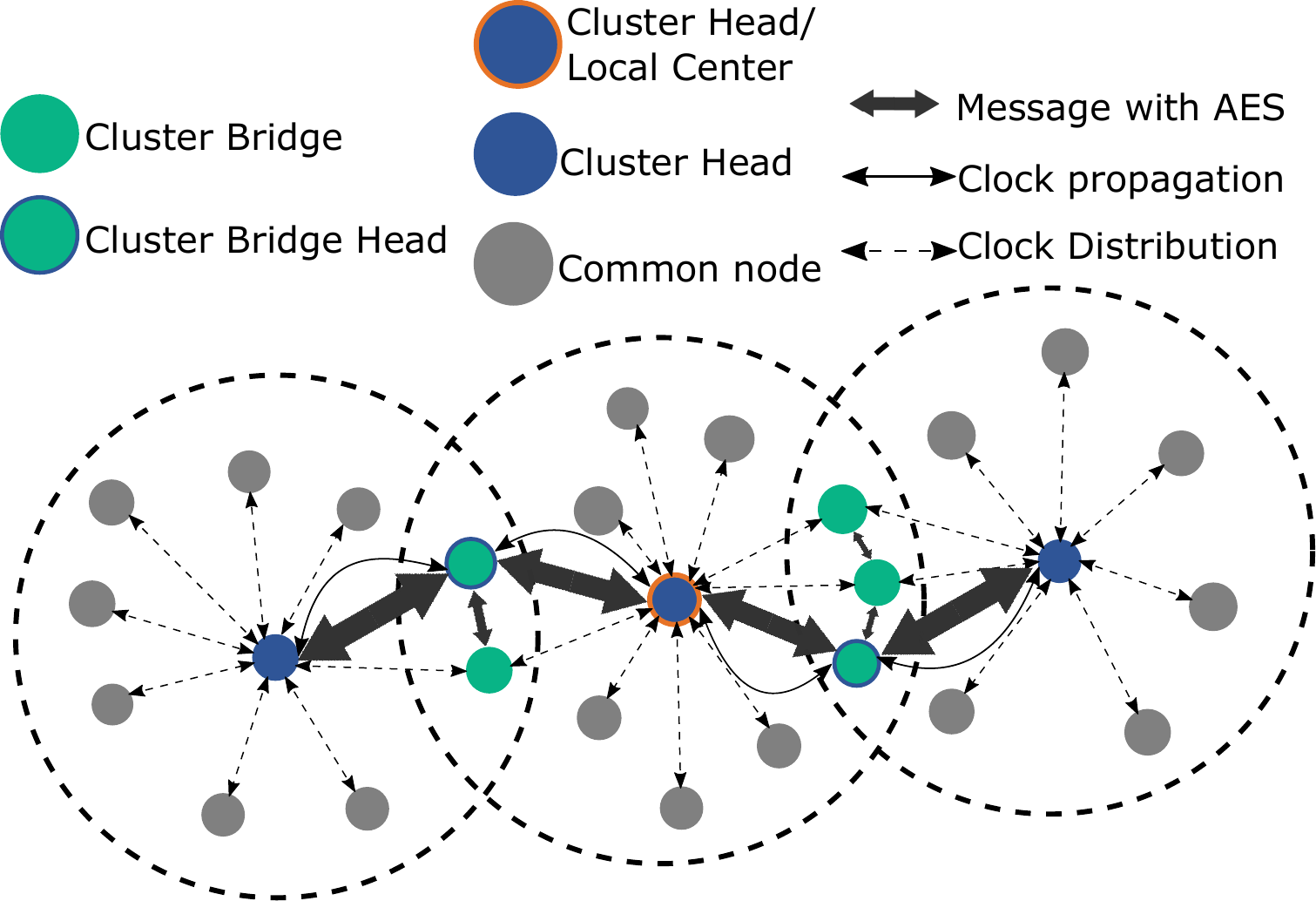}
	\caption{Time Synchronization in C-sync using a clustered architecture for power efficiency and fault resilience.}
	\label{general_nw}
\end{figure}

A faulty node with incorrect information could influence the network in existing synchronization solutions, albeit providing the basic deterrence against faults since there is no verification of the data. For example, a critical sensing application with distributed logging of events may be rendered useless if an event's timestamp is recorded differently by a few devices in the neighborhood of the event. Erroneous information disrupts the analysis and debugging of data. Hence, there is a need to integrate fault-tolerance while designing a synchronization protocol to ensure that the error from the time source/connecting node does not impact the entire network.

Some of the existing protocols handle node failures by switching to a different node to provide a reference clock~\cite{Maroti,6777583,5779066}.
Although reference node switching could work as a tentative solution, an adequate fault-handling mechanism is absent. For example, faults such as selective forwarding in a reference node could send all but critical messages resulting in desynchronization without any means to fix the issue. 
Generally, protocols reliant on a reference node require additional measures to mitigate the impact of a compromised reference node. 
In the absence of reference nodes, some protocols adopt a decentralized design that has inherent fault resilience~\cite{5211944,7504162,6926769,6520853}. Faulty nodes are excluded from information transmission during synchronization if their information is substantially different from those of other neighboring nodes. However, managing the faulty nodes in a decentralized network exponentially increases the messages exchanged and consequently, the power consumption.



Typically, time synchronization protocols focus on three primary goals: energy efficiency, accuracy and scalability. To achieve energy efficiency, these solutions minimize communication (radio ``ON" time) such that the least number of messages are exchanged to achieve and maintain synchronization. Synchronization accuracy varies from a few seconds~\cite{103043} to a few microseconds~\cite{5779066,6777583,5211944}, depending on the type of protocol and application used. For sensors in real-time systems, the expected accuracy is in the order of microseconds~\cite{Ganeriwal:2003:TPS:958491.958508}. 
Typically, achievable accuracy is bounded by the resolution and the stability of the clock used by the node. 
Scalability ensures growth in network size does not impact the functionality and the performance of the protocol. 

In this paper, we propose a decentralized clustering-based time synchronization protocol as shown in Figure~\ref{general_nw} (referred to as \emph{C-sync}) to ensure fault resilience of the network in addition to the standard three metrics.  
A representative node from each cluster - Cluster Head (CH) is connected to other clusters through a few Cluster Bridge (CB) nodes for information transmission. 
As CH and CB nodes wield greater influence on the propagated information across clusters, we design a consensus mechanism among the nodes of the neighborhood to verify the information correctness and integrity from CH/CB.
More information on the fault model and the fault handling in C-sync is explained in Section~\ref{fault_model}.

In C-sync, a concept called Local Centers (LC) is introduced to handle scalability, where some CH nodes within the network are elected as reference nodes.
These nodes coordinate the distribution of time information such that the synchronization error is limited by restricting the number of hops between LCs and other nodes.
The resilient design of C-sync coupled with low power consumption and scalability enables design of real-time applications in decentralized systems.
LCs are further explained in detail in Section~\ref{consensus}. 
To summarize, the contributions of this paper are as follows:
\begin{enumerate}
	\item We propose C-sync, a decentralized fault-resilient clustering-based time synchronization protocol suited for large-scale IoT networks. C-sync introduces multiple time sources in the network to constrain the error between any node and its reference. 
	\item We show that the proposed protocol achieves synchronization with significantly lower power consumption while ensuring accuracy is not compromised.
	\item Through extensive experiments on a real testbed and theoretical analysis, we show the fault recovery mechanism of C-sync and show the performance with power efficiency and synchronization accuracy.
\end{enumerate}

\paragraph{Organization}

A detailed description of the C-sync protocol is provided in Section~\ref{protocol}.
Section~\ref{fault_model} outlines the considered fault model and the fault-handling mechanism of C-sync.
The experimental setup and the experimental results are discussed in Section~\ref{exp}.
A comparison of the existing time synchronization solutions is presented in Section~\ref{literature}.
Section~\ref{conclusions} concludes this paper with some future research directions.




\section{C-Sync}\label{protocol}
In this section, we introduce our protocol C-sync and discuss the synchronization mechanism. 
The pseudo-code of C-sync protocol operations is shown in Algorithms~\ref{alg:csync1} and~\ref{alg:csync2}.

\begin{figure}[ht]
	\centering  
	\includegraphics[width=0.95\linewidth]{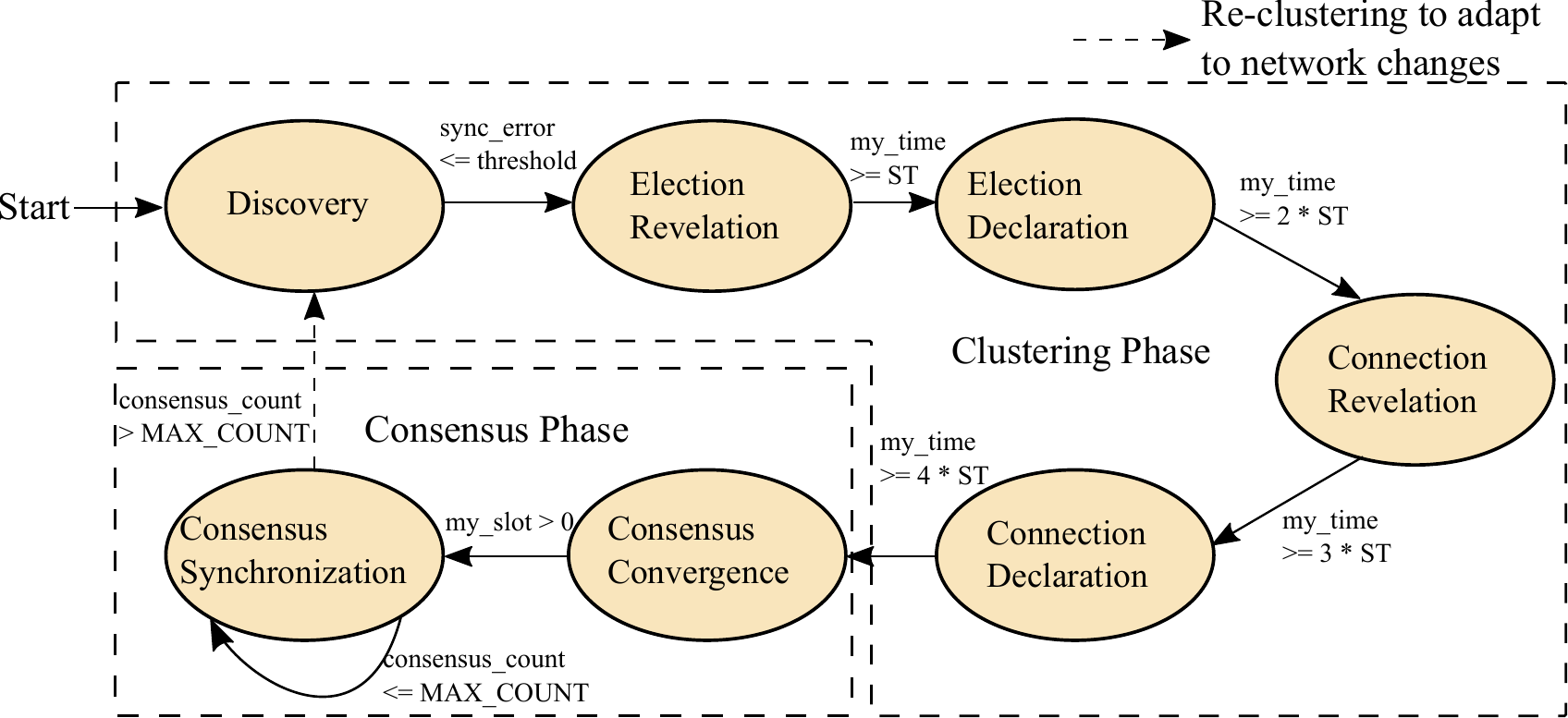}
	\caption{The two phases of the C-sync protocol represented as a state machine.}
	\label{fig:state_machine}
\end{figure}

\subsection{Background}

The C-sync protocol follows a 2-phase process as described in the state machine in Figure~\ref{fig:state_machine}.
Clustering is the first phase of C-sync derived from the existing clustering scheme DeCoRIC \cite{9209755}. 
DeCoRIC uses a clustering scheme based on the degree of a node, \emph{i.e.,} the node with the highest degree forms the representative CH node that facilitates routing information within and across clusters.
The degree of a node is the number of active communication links of a node with its neighbors.
DeCoRIC is adapted to the synchronization process to establish the underlying architecture in the network.
The state machine from DeCoRIC is transformed with additional states to integrate scheduling between the states. 
Once the clusters are formed, the consensus phase maintains synchronization among the clusters. 
Since we assume an ad-hoc network with the same wireless channel, any node that joins the network in the midst of the state machine operation is able to join the nearest cluster and associate its time to that of the CH.
The new node may become a CH/CB node depending on its position in the subsequent clustering phase.

\subsection{Clustering}\label{clustering}
The clustering phase comprises five states to establish the clustered architecture among the nodes.
\paragraph{Discovery state.}
In the discovery state, each node broadcasts messages with the time information derived from its hardware clock and listens to other nodes in the neighborhood to discover its environment.
In this state, the logical (\emph{i.e.,} global) clock is the same as the hardware clock and the degree of every node is set to $0$.
The logical clock $L(t)$ and its parameters (rate and offset) is derived from the hardware clock value $h(t)$ compensated by the rate parameter and an offset value as in~\cite{5211944}:
\begin{equation}\label{eq:abs_log_clk}
L(t) = \int_{\tau=0}^{t} h(\tau)l(\tau) \,d\tau + \theta(t),
\end{equation}
where $l(\tau)$ is the logical clock rate relative to hardware clock (\emph{logical clock rate} in short) and $\theta(t)$ is the logical clock offset. The logical clock rate is the average of relative rates to all its neighbors, whereas the logical offset is the average of relative offsets to all its neighbors.

For a node $i$ with $n_i$ neighbors, the clock rate $l_i$ and offset $\theta_i$ is computed using

\begin{equation}\label{eq:rel_clk_rate}
l_i(t+1) = \frac{\sum_{j \in n_i} l_j(t) + l_i(t)}{n_i + 1} ~~~~~~~~\mathrm{and}
\end{equation}

\begin{equation}\label{eq:rel_clk_offset}
\theta_i(t+1) =  \theta_i(t) + \frac{\sum_{j \in n_i} L_j(t) - L_i(t)}{n_i + 1}.
\end{equation}


Upon reception of a message from a new neighbor, the receiving node increments its degree and stores the time information (rate and offset) from both logical and hardware clocks of the sender.
All the messages in this state are asynchronous since there is no reference clock.
The MAC layer uses CSMA-CA with message timestamps to minimize the collision and interference in the network as the nodes are asynchronous at the start.
The logical clock rate and offset of a node are derived from averaged relative clock rates and offsets of its neighbors.
Since all nodes perform averaging for the offset and rate, the nodes are loosely synchronized.
The average network delay gets factored into the offset as the logical clock offset gets calculated independently at each node based on the neighbor offsets.
If a node can achieve both the offset and rate of its logical clock within a pre-defined threshold, that node creates a reference time (delay) for a state transition to the election revelation state for all nodes in the neighborhood.
The threshold is set based on worst-case CSMA backoff to ensure that nodes are able to receive messages even if they are loosely synchronized.
The threshold also allows nodes with large offsets to jump to the transition interval based on a synchronized neighbor, preventing long convergence times.
Lines 1-18 of Algorithm~\ref{alg:csync1} show the operation of the discovery state and the transition to the election revelation state.

In the event of a late-discovered node due to the dynamic nature (new/mobile nodes) of the network, the new node listens to messages from the neighboring nodes and associates itself with the nearest CH node as a Common node (CM) of its cluster.
Upon receiving the state transition announcement, the CM node directly jumps to the current state of the network.
However, it is important to note that the subsequent discovery state may elevate the status of the CM node to a CB or a CH depending on its position in the network to form more optimal clusters.

\paragraph{Election revelation state.}
The election revelation state facilitates the exchange of degree information from all nodes for the election of Cluster Heads (CHs). 
The configuration could be changed to use residual energy of the nodes to ensure uniform distribution of energy among the cluster nodes.
CH nodes retain the CM nodes in sleep mode except during transmission, thereby significantly reducing energy.
Beyond the creation of reference time from the previous state, all nodes transition to further states at a pre-defined interval of time called the \emph{state transition} (ST) interval.
The ST interval is a combination of the error threshold used in the discovery phase and frame length to ensure reduction in the wireless channel interference and a successful message transmission.
All nodes transition to the election declaration state after the ST interval.

Note that the slotted ST intervals are based only on the synchronized time without any dependency on the channel parameters.
Additionally, the ST intervals are generated based on the hardware timers of a node without any dependency on software time slots.
Although integrating channel information would allow the use of C-sync in Time Slotted Channel Hopping (TSCH) networks, we focus on time synchronization on the same channel for all nodes. Use of the same channel enables a plug-and-play interface with greater network flexibility in contrast to configuring every node.

\paragraph{Election declaration state.}
Nodes with the highest degree declare their CH status and form clusters.
If a node receives a message with a higher degree from its neighbor, it cancels its broadcast and associates itself to the CH node as a common node (CM).
Operations of election revelation and election declaration states are shown in lines 19-31 of Algorithm~\ref{alg:csync1}.

\paragraph{Connection revelation state.}
Connectivity among the clusters is established through the election of \emph{Cluster Bridge} (CB) nodes that connect two or more CH nodes.
CB nodes play an important role to ensure fault-free dissemination of messages during the consensus phase.
Hence, to authenticate and prevent nodes to falsely declare themselves as CB, the messages in this state use AES encryption with the IDs of the CH nodes as the key as shown in Figure~\ref{general_nw}.
This encryption is in addition to the already encrypted messages of IEEE 802.15.4 standard to ensure the authenticity of CB nodes~\cite{1700009}.
The ST interval transitions the network to the connection declaration state.

\setlength{\dbltextfloatsep}{0pt}
\begin{algorithm}
	\scalebox{0.90}{
		\begin{minipage}{0.98\columnwidth}
			\caption{Pseudo-code representing the operation of Clustering phase in C-sync.}
			\label{alg:csync1}
			\begin{algorithmic}[1]
				\State INITIALIZE $\var{Neighbors} \gets \emptyset$ 
				\State $\var{my\_state} \gets \var{DISCOVERY}$
				\State $\var{Msg.Id} \gets \var{my\_addr}; \var{Msg.time} \gets \var{my\_time}$
				\State $\var{my\_CH\_count} \gets 0; \var{consensus\_count} \gets 0$
				\State $broadcast(\var{Msg})$
				\If{$rcv()$}
				\For {$\var{n}$ in $\var{Neighbor.list()}$}
				\State $\var{n.time} \gets rcv().\var{time}$ 
				\If {$\var{n}.synced()$}
				\State $\var{my\_time} \gets \var{n.time}$
				\State $\var{my\_state} \gets \var{ELECTION\_REVELATION}$
				\EndIf
				\EndFor
				\If{$\var{n} \notin \var{Neighbor.list()}$}
				\State $\var{my\_degree} \gets \var{my\_degree} + 1$ 
				\State $\var{Neighbor.list()}.Append(n)$
				\EndIf
				\EndIf
				\If {$(\var{my\_time} \ge ST)$ \AND $(rcv())$}
				\State $\var{Neighbor.degree} \gets rcv().\var{degree}$
				\State $\var{my\_state} \gets \var{ELECTION\_DECLARATION}$
				\EndIf
				\If {$(\var{my\_time} \ge 2 \cdot ST)$ \AND $(rcv())$}
				\If {$(\var{my\_degree} > rcv().\var{degree})$ \OR $(\var{my\_degree} == rcv().\var{degree})$ \AND $(\var{my\_addr} > rcv().\var{addr})$}
				\State $\var{my\_role} \gets CH$
				\Else
				\State $\var{my\_CH\_count} \gets \var{my\_CH\_count} + 1$
				\State $\var{my\_role} \gets CM$
				\EndIf
				\State $\var{my\_state} \gets \var{CONNECTION\_REVELATION}$
				\EndIf
				\If {$(\var{my\_time} \ge 3 \cdot ST)$ \AND $(rcv())$}
				\For {$\var{n}$ in $\var{Neighbor.list()}$}
				\If {$(\var{n.role} == CH)$ \AND $(\var{my\_CH\_count} \le 2)$}
				\State $\var{my\_CH\_count} \gets \var{my\_CH\_count} + 1$
				\State $\var{my\_CH\_list} \gets \var{n}$
				\Else
				\State $\var{my\_role} \gets CB$
				\EndIf
				\EndFor
				\EndIf
				\State $\var{my\_state} \gets \var{CONNECTION\_DECLARATION}$
				\If {$(\var{my\_time} \ge 4 \cdot ST)$ \AND $(rcv())$}
				\If { $(\var{my\_degree} > rcv().\var{degree})$ \OR $(\var{my\_degree} == rcv().\var{degree})$ \AND $(\var{my\_addr} > rcv().\var{addr})$}
				\State $\var{my\_role} \gets CBH$
				\EndIf
				\State $\var{my\_state} \gets \var{CONSENSUS\_CONVERGENCE}$
				\EndIf
				\Comment{//End of clustering}
			\end{algorithmic}
	\end{minipage}}
\end{algorithm}

\paragraph{Connection declaration state.}
Similar to the election declaration, CB nodes (CM nodes with multiple CH connections) declare their status to the CH nodes in the connection declaration state. 
As with the CH election, CBs with the highest degree or highest address (in case of the same degree) declare themselves as the representative Cluster Bridge Head (CBH) and disclose their neighboring CH nodes.
The CH nodes discover their neighboring CH nodes through their CB(s). 
The steps involved in connection revelation and connection declaration states are shown in lines 32-48 of Algorithm~\ref{alg:csync1} with a transition to the consensus phase after the ST interval.

It is important to note that the consensus phase starts with a new reference time (common across the network) and less contention on the network as message transmission is restricted only to CHs and CBs to finalize the time slots.
The CM nodes listen to the information from CH transmissions and update their clocks.
All nodes transmit in the discovery phase continuously while only the CH and CB nodes transmit beyond the election phase in a time-slotted manner as the synchronization among nodes improves along the state transitions.
In contrast, most flooding-based protocols experience channel contention among nodes leading to significant packet and energy losses.

The topology/status of all nodes (CH/CB/CM) of the clustered architecture remains intact for the consensus phases. 
However, any network changes (node mobility or unstable communication links) are reflected in the subsequent discovery phase, where different nodes may get elected as CH/CB depending on the change.
The discovery phase gets triggered after the configured number of repetitions (MAX\_COUNT) of the consensus synchronization phase are complete. The configuration is set based on the frequency of changes in the network, \emph{i.e.,} dynamic networks with higher node failures/mobility have shorter repetitions and vice-versa.
New nodes become CMs and listen to the nearest CH for the time information while missing nodes are treated as fail-stop faults and handled according to the steps discussed in Section~\ref{fault_model}.
Exclusive discussion on the mobility of nodes is beyond the scope of this paper.



\begin{algorithm}
	\scalebox{0.90}{
		\begin{minipage}{0.98\columnwidth}
			\caption{Pseudo-code representing the operation of Consensus phase in C-sync.}
			\label{alg:csync2}
			\begin{algorithmic}[1]
				\For {$\var{slot}$ in $MAX\_SLOTS$}
				\For {$\var{CH\_neighb}$ in $\var{CH\_neighbors}$}
				\If{$\var{my\_CH\_degree} \le \var{CH\_neighb}$}
				\State $\var{my\_slot} \gets \var{slot}$
				\ElsIf {$\var{slot} == MAX\_SLOTS$}
				\State $\var{my\_slot} \gets MAX\_SLOTS$
				\EndIf
				\State $\var{my\_state} \gets \var{CONSENSUS\_SYNCHRONIZA-}$\
				 $\var{-TION}$
				\EndFor
				\EndFor
				\State $\var{my\_slot} = MAX\_SLOTS - \var{my\_slot}$
				\While {$\var{consensus\_count} \le \var{MAX\_COUNT}$}
				\For {$\var{slot}$ in $MAX\_SLOTS+1:2\cdot MAX\_SLOTS$}
				\If {$\var{my\_slot} == \var{slot}$}
				\State $ increment(\var{consensus\_count})$ 
				\State $\var{my\_time} \gets rcv().\var{time}$
				\If {$\var{my\_role} == CB$ \OR $CH$}
				\State $broadcast(\var{Msg})$
				\EndIf
				\EndIf
				\EndFor
				\EndWhile
				\State $\var{my\_state} \gets \var{DISCOVERY}$
				\Comment{//End of consensus}
			\end{algorithmic}
		\end{minipage}}
	\end{algorithm}

\subsection{Consensus}\label{consensus}
The consensus phase is a stable repetitive phase after clustering for maintenance of synchronization and the clustered architecture.
There are two states in this phase: consensus convergence and consensus synchronization.

\paragraph{Consensus convergence state.}
In this state, a pre-configured number of time slots (MAX\_SLOTS) are available for the CH nodes to transmit their messages. 
The slot of a CH is decided based on the number of its CH neighbors obtained in the connection declaration state.
For example, a CH with only one CH neighbor node transmits first, followed by CHs with two neighboring CH nodes and so on.
Generally, CH nodes with a lower number of CH neighbors (one or two CH neighbors) tend to be at the edge of the network while CH nodes with a higher number of CH neighbors (two or more CH neighbors) are located towards the center of the network.
This fact is exploited by our protocol to find one or more local center (LC) nodes depending on the size of the network that can act as a time source to synchronize different parts of the network, leading to localized time distribution.
LCs are CH nodes that have the highest neighboring CH connections and are typically located towards the center of the network.
Lines 1-11 describe the consensus convergence state in Algorithm~\ref{alg:csync2}.

The number of time slots is proportional to the number of hops a node at the edge of the network traverses to an LC node.
In the case of special network topologies like a chain or a ring, where most nodes have the same set of CH neighbors, the configurable time slots limit the number of hops to reach the LC node.
If a CB receives a message from a CH, it acknowledges the message confirming its slot.
Any CB that receives two different time slots from neighboring CHs chooses the higher slot number to acknowledge by convention.
Similarly, a CH node updates its time slot to a higher value if its neighbor slot is higher than its initial slot (in the case of chain/ring topologies).
The transmission continues until all the CH nodes have a confirmed time slot. 
Similar to slot selection, if there are multiple nodes with the same number of neighboring CH and time slots, the node with a higher ID is chosen as the LC.
Multiple LCs are found depending on the size of the network and these nodes disperse the time information back to the CH nodes.


A notable caveat is that the minimization of hops is adaptive to the network and the topology, \emph{i.e.,} a network could have multiple LCs (within a distance of 1-2 hops) or a single LC (configured maximum hops).
This novel method of limiting the number of hops to the time source achieves a simplistic solution eliminating the requirement of any additional hardware/timing adjustment. 
Although existing solutions are able to limit the error significantly, it is functional only up to a certain number of hops and the problem repeats upon further scaling the network.

\paragraph{Consensus synchronization state.}
In this state, LC transmits time information to the CH nodes and further to the CM nodes of their respective clusters.
The time slots for this state are the modulo time slot proportionate to the pre-configured time slots, \emph{i.e.,} if a CH node transmitted at slot 4 in a 10-time slot window during convergence, it receives its time information at slot 6 (10-4) during synchronization.
CM nodes are awake to receive their slot numbers from CH (same as their associated CH's slot) and go into sleep mode.
Since the time information is passed from the local center, the clock rate and offset are updated relative to the LC. 
This operation is shown in lines 12-23 of Algorithm~\ref{alg:csync2}.

Synchronization errors accumulate at every node starting from the LC until the CM nodes along the path. 
As practical clocks have variations in both offset and drift, both parameters need to be compensated.
The logical clock rate of a receiving node (r), from Equation~\eqref{eq:abs_log_clk}, is defined as the ratio of a logical clock (global clock value) of a node to its hardware clock after offset compensation and is given by:
\begin{equation}
l_r(t) = L_r(t)/h_r(t).
\end{equation}

Since the rate is dependent on the hardware clocks of each node along the multi-hop path, it is important to adjust the rate only to the logical clock of the LC as the reference clock. 
The relative clock rate of receiving node (r) relative to the sending node (s) is the ratio of the logical clock of the sender ($L_s(t)$) to the hardware clock of the receiver ($h_r(t)$), given as:
\begin{equation}
l_{rs}(t) = \frac{L_s(t)}{h_r(t)}.
\end{equation}

If a node is directly connected to the LC, the relative rate would be sufficient to compute the logical clock.
For non-direct neighbors, we compute the logical clock rate with reference to LC $(l_{LC})$ as the ratio of the relative rate to the current clock rate of the receiver:

\begin{equation}
l_{LC} = l_{rs}(t)/l_r(t).
\end{equation}




CH nodes turn on only in their respective slots while CB nodes remain active for two slots to receive and send information to their CH neighbors respectively.
The clock rates of two LCs are averaged to establish a uniform clock synchronization across the network if a CH node in any path is also an LC.
The nodes move into an idle phase after synchronization where they are in sleep mode and no messages are being exchanged.
C-sync switches to consensus synchronization periodically for LC to distribute time information for maintaining synchronization.


\paragraph{C-sync Overhead}
Existing synchronization solutions have an overhead to cover the entire network diameter with a lack of backup mechanisms to handle a dynamically changing network. 
On the contrary, the diameter in C-sync ranges from a single cluster to distance (hops) to the LC, yielding a much lower overhead. 
C-sync employs re-clustering to ensure an efficient clustered architecture and has a fault detection and correction mechanism in place to handle changes in the network.
Additionally, since the protocol starts with a completely decentralized network, the computational complexity of each node of the protocol is $O(n)$.

\begin{lemma}\label{bounded_error_defn}
	The maximum synchronization error between any node to its nearest LC is a parametric value.
\end{lemma}

\begin{proof}[\unskip\nopunct]
The synchronization error in most time synchronization protocols is dependent on the propagation time, frequency of messages exchanged and the number of hops required to communicate.
Due to MAC-layer time-stamping, the propagation time can be safely ignored assuming no channel contention and interference.
This is because the propagation delay roughly amounts to \SI{0.3}{\micro\second} for 100m distance between the nodes while the resolution of hardware timer is about \SI{1.9}{\micro\second}.
It is important to note that this delay applies to the consensus synchronization state.
If the delay between two consecutive messages exchanged is $\tau$ and the minimum achievable synchronization error for an ideal "zero"-delay is $\delta$, the accumulated error due to delay in message exchange is $\tau \cdot \delta$.
In C-sync, $\tau$ is the idle-time delay between two synchronization messages, and the maximum number of hops to an LC is represented by $\eta$.
As the synchronization error increases at each hop, the total error from any node to its LC is given by $\eta \cdot \tau \cdot \delta$.
This parametric limit restricts the synchronization error for any node in the network to its LC.
\end{proof}

\section{Resilience in C-sync}\label{fault_model}
In this section, we categorize the faults and list the assumptions made by C-sync for fault handling. 
Further, we prove by induction that any fault in our described categories can be handled by C-sync if the assumptions are met. 



\paragraph{Fault Model}In this paper, we consider two types of faults: fail-stop faults and a subset of byzantine faults. 
With a fail-stop fault, the node is non-responsive due to battery exhaustion, hardware/software damage and/or environmental factors, etc.
We also consider byzantine faulty nodes where the nodes could behave erratically making it difficult and expensive in terms of communication to detect them. 
The subset of byzantine faults considered in this paper includes spikes, outliers and intermittent communication faults~\cite{fault_survey,6489879,KARLOF2003293}.
A Spike fault is a sudden surge in reported values that may or may not subsequently return to normal values.
When the reported values are beyond the boundary of the expected values, the resulting fault is an outlier.
With intermittent faults, the message transmission from a node is sporadic with periods of inactivity.

Typically, in cluster-based network architectures, the faults described above can be translated as selective forwarding, discovery flooding and altered information~\cite{KARLOF2003293}.
Fail-stop faults and intermittent communication mimic a potential temporary loss of communication resulting in selective forwarding.
Additionally, nodes may send sudden variation (spikes, outliers) in the information due to a fault (E.g. False perception of the environment, routing changes, etc. ) leading to altered information. 
A threshold on the acceptable range for the received data can prevent the altered information faults such as spikes and outliers. 
A node with a faulty radio could end up in a high transmit power resulting in discovery floods (also called HELLO floods). 
The high-powered transmission could lead to the false election of these nodes as one of the critical routing nodes (CB or CH).
For discovery floods, it is important to verify the bi-directionality of the links between the nodes, \emph{i.e.,} to verify that the link has the same properties in both directions. 
The above faults are the observable faults through radio communication and constitute a subset of the generic byzantine faults.
Additionally, we assume that the faulty nodes can cooperate with each other independent of their node type.
The impact of the faults observed in both phases of C-sync described in Section~\ref{protocol} is studied.




\paragraph{Impact of faults} 
A faulty node can report an erroneous degree and/or use discovery flooding with a high-powered (damaged) radio to become a CH/CB node during the clustering phase.
As a routing node, it could have a wider impact during the consensus phase transmitting erroneous time information across clusters.
Since CB plays a critical role in the fault detection and correction process, an additional authentication using AES cipher is used with a combination of communicating CH nodes' ID as the key to prevent non-neighboring nodes of CHs from getting promoted to a CB.
Replicating a MAC address of 8 bytes in the ID through brute force is a highly difficult and energy-intensive task for a resource-constrained node.
The AES cipher is used atop the existing message encryption mechanism of IEEE 802.15.4~\cite{1700009} to authenticate the CB nodes and verify the bi-directionality of the links~\cite{KARLOF2003293}.
Authentication of CB nodes reduces the number of nodes participating in the byzantine consensus (agreement) among the correct nodes, yielding significant energy savings~\cite{yin2003separating}.
All the other faulty nodes (CH/CB/CM) are detected and corrected in the consensus phase using the Byzantine consensus mechanism. 

\subsection{Assumptions}
Let $n_i$ denote the number of neighbors of any node $i$ in a cluster within its communication range, $\mathrm{n_{CB}}$ denotes the set of CB nodes between two clusters with cluster heads CH1 and CH2 such that CH1, ${n_{CB}}$ and CH2 are at an increasing number of hops from the local center respectively. The assumptions and broadcast primitives required for byzantine consensus are:
\begin{assumptions}\label{assum:no_fake_addr}
	 No node can fake its address or the reference address in a message as it is a hardware-based MAC address. 
\end{assumptions}
\begin{assumptions}\label{assum:min_neighbors}
	Every node in a cluster has at least $\floor{\frac{n_i}{2}} + 1$ neighbors that are fault-free.
\end{assumptions}
\begin{assumptions}\label{assum:min_CBs}
	There are at least $\floor{\frac{n_{CB}}{2}} + 1$ CB nodes that are fault-free between any two clusters. 
\end{assumptions}
\begin{assumptions}\label{assum:min_cluster}
	There are at least two clusters connected by $\floor{\frac{n_{CB}}{2}} + 1$ CB nodes without any network partitions/isolated clusters.
\end{assumptions}

\subsection*{Atomic Broadcasts}
Typical byzantine agreements require significant energy to perform computations and communication. Sensor nodes are resource-constrained and require an energy-efficient way to achieve a byzantine agreement.
Atomic broadcasts were introduced in~\cite{CRISTIAN1995158} to achieve byzantine consensus if they meet the following criteria:
\begin{itemize}
	\item Every message from a correct sender is received by all correct receivers within a time-bound.
	\item Every message is received by the correct receiver in the same order as it was sent by the correct sender.
\end{itemize}
The byzantine agreement is concluded if all the correct nodes have the correct information that was propagated.

In C-sync, atomic broadcasts are initiated by the CB nodes if the CBH or CH nodes are non-responsive or send incorrect information.
To understand further, we take a look at both types of faults and present the fault handling mechanisms.

\subsection{Fault detection and correction}\label{sec:fault_handling}
With C-sync, it is expected that the CH and CBH nodes' synchronization messages are transmitted at a scheduled time slot. 
Any missing broadcast with selective forwarding or fault time information due to byzantine fault is handled through byzantine consensus among nodes of the cluster.

All CBH and CB nodes schedule a byzantine consensus message at a delay of two message transmissions time after the transmission of the synchronization message and monitor the information within the cluster to verify its authenticity. 
Two transmission times account for message transmission from CBH to CH and, further, from the CH to all the cluster nodes.
The CB nodes drop the scheduled message upon correct information from CBH and CH nodes in their respective time slots, while a missed/incorrect message triggers the transmission of the byzantine consensus message as per schedule.
Note that a consensus message can be configured to trigger after a certain number of repetitive misses/violations.
The detected faulty node is added to a blacklist where messages originating from the blacklisted nodes are ignored by the other nodes and are excluded from forming new clusters.
Upon receiving the atomic broadcast, all the CM nodes and/or CH nodes (non-faulty) re-transmit this information till every node in the cluster receives at least $\floor{\frac{n_i}{2}} + 1$ correct messages. 
The above condition is both \emph{necessary} and \emph{sufficient} to achieve a successful byzantine consensus (agreement).

A faulty CM node can only assert its influence by sending a false byzantine consensus message. 
However, this message is not replicated by all of the other CM nodes as it was not sent by any of the CB nodes directly/indirectly (as reference).
Thus, the minimum number of messages will not be received by any of the other nodes of the cluster, leading to the failure of atomic broadcasts.
The clustered architecture and byzantine consensus mechanism established in C-sync can handle all the aforementioned fault types and can further be extended to handle other fault types and even certain attacks. The discussion on these extensions are beyond the scope of this paper. 


The flooding of messages within the cluster leads to a temporary increase in power consumption. 
However, it is a small price to pay to contain faults within a cluster and achieve resilience against faulty nodes.
It is important to note that the byzantine consensus completes within the same time slot when the time information has to be distributed and hence, does not impact accuracy.

\begin{thm}\label{fault_thm}
	Any faulty node in the network can be detected and corrected in C-sync if the stated assumptions are satisfied.
\end{thm}

\begin{proof}
	
We prove the theorem using the principle of induction. 
The proof is provided with reference to a single cluster as a representative case consisting of a cluster head CH  that connects to other clusters through CB nodes and the associated CM nodes such that every node meets Assumption~\ref{assum:min_neighbors}.
Let $\mathrm{CB1}$ and $\mathrm{CB2}$ denote the set of CB nodes of CH, with $\mathrm{CB1}$ located closer to the LC and $\mathrm{CB2}$ located farther away from LC, such that the information chain traverses from LC to CB1 to CH to CB2.
Multiple nodes are present within the sets CB1 and CB2 to ensure that Assumption~\ref{assum:min_CBs} is met.
Without loss of generality, the same proof applies to every cluster in the network independently.
Also, the proof for the CB1 set applies to the CB2 set as well.

For the base case, let us assume there is only one fault in the cluster.
If CH is the faulty node, the information from CB1 is either dropped or modified before distributing it to CB2 and its CM nodes.
CB1 nodes including CB1H (CBH in the set of CB1) send a byzantine consensus message immediately (modified time information) or as scheduled (selective forwarding) such that it gets re-transmitted across the cluster upon reception. 
Based on Assumptions~\ref{assum:min_neighbors} and~\ref{assum:min_CBs}, the fault-free $\floor{\frac{n_i}{2}} + 1$ neighbors propagate the correct information from the received byzantine consensus messages to reach an agreement. 
Similarly, if the CB1H node or one of the CB1 nodes skip sending the synchronization message or send faulty time information, the remaining CB1 nodes initiate the byzantine consensus.
A faulty CM node can directly trigger a byzantine consensus message with incorrect time information.
However, since the CB1H address cannot be used as a reference, the fault-free neighbors do not re-transmit this message.
By Assumption~\ref{assum:min_neighbors}, byzantine consensus for the faulty CM node will not be reached.

Let us assume that the C-sync protocol can detect and correct up to $k$ faults that satisfy assumptions~\ref{assum:no_fake_addr}, \ref{assum:min_neighbors} and \ref{assum:min_CBs}, \emph{i.e.,} $k < \frac{n_i}{2}$ for all the cluster nodes including CH and $k < \frac{n_{CB1}}{2}$ for the set CB1.

Consider $k+1$ faults that satisfy the assumptions~\ref{assum:no_fake_addr}, \ref{assum:min_neighbors} and \ref{assum:min_CBs} and the assumption of $k$ faults as stated above.
If the additional faulty node in the $k+1$ faults is a CM node and assumption 1 holds, this CM node can trigger a byzantine consensus message without CBH or any of CB1 as reference nodes.
Additionally, with assumption~\ref{assum:min_neighbors}, we can infer that $(k+1) < \frac{n_i}{2}$.
Thus, the fault will not be propagated and consensus on faulty information will not be reached \emph{i.e.,} correct nodes are not impacted.
With CH as the faulty node, CB1 nodes initiate byzantine consensus messages immediately if the error from the CH node exceeds the threshold.
Since $(k+1) < \frac{n_{CB1}}{2}$ and $(k+1) < \frac{n_i}{2}$ hold, then the messages with correct information from the $\frac{n_i}{2} + 1$ nodes (CB1 and CM) of the cluster dominate the faulty byzantine messages ($k+1$) leading to successful detection and correction of the fault.
Lastly, if either CB1H or any of the other CB1 nodes is the additional faulty node, and $(k+1) < \frac{n_{CB1}}{2}$ holds, the messages from the $k+1$ nodes are not sufficient to reach the byzantine agreement.
Thus, the correct CB nodes are able to disseminate the information within the cluster.

Theorem~\ref{fault_thm} holds for any $k$ such that assumptions~\ref{assum:no_fake_addr},~\ref{assum:min_neighbors} and~\ref{assum:min_CBs} are satisfied.
Hence, with the C-sync protocol, wireless networks can achieve resilience to faults.
\end{proof}

\subsubsection*{Discussion}
\label{sec:exceptions}
An exception to the fault handling in C-sync is the case of node mobility during message transmission.
A mobile node could initiate a transition from one cluster to another cluster during the synchronization message transmission (consensus synchronization state), mimicking a fail-stop fault/selective forwarding.
Consequently, the node gets blacklisted although movement was a legitimate action.
This exception can be addressed through an exchange of membership information in addition to time information among different CH nodes of the network in the consensus convergence state and is beyond the scope of this paper.
Additionally, small networks that result in a single cluster without a CB are also an exception to the fault handling mechanism of C-sync. 
To handle faults in this scenario, the CM nodes would be awake for a longer duration to initiate a consensus flooding (if required), resulting in a power efficiency vs. fault tolerance trade-off.

\section{Experiments and Results}\label{exp}


To test the performance of the C-sync protocol, we introduce a fault in a CH node of a simple representative network to show that the fault detection, correction and containment within the cluster are achieved. 
Further, we conduct experiments to measure the power consumption as well as the accuracy of synchronization. 
Finally, we demonstrate through a chain topology that the synchronization error for any node to its nearest time source is restricted.

\subsection{Experimental Setup}
\label{sec:setup}
The C-sync protocol was implemented on the Contiki 3.0 operating system~\cite{Dunkels}, with Tmote Sky~\cite{skymote} boards utilizing the IEEE 802.15.4 communication standard.
While nodes with the same capabilities are used for testing purposes, C-sync can be used in networks having nodes with heterogeneous capabilities.
Tmote Sky consists of MSP430~\cite{msp430} micro-controller and CC2420~\cite{cc2420} radio communication chip from Texas Instruments.
The experiments were conducted on the Indriya~\cite{indriya} testbed at the National University of Singapore, where there are over 50 Tmote Skys deployed on different floors of a building. 
The implementation has been made available for verification.~\footnote{https://github.com/nitinshivaraman/C-sync}

The software architecture of C-sync is shown in Figure~\ref{fig:contiki}.
Starting from the lowest layer, the network stack maps the bottom three layers of the OSI protocol stack whose output is passed to the upper layers.
The network stack handles the radio communication including physical layer control, link-layer security features such as AES, CSMA-CA, and MAC-layer features such as MAC-layer timestamping.
Chameleon is a header transformation layer that adds or removes the header component from the packet buffer.

 \begin{figure}
	\centering
	\includegraphics[scale=0.8, width = 0.94\columnwidth]{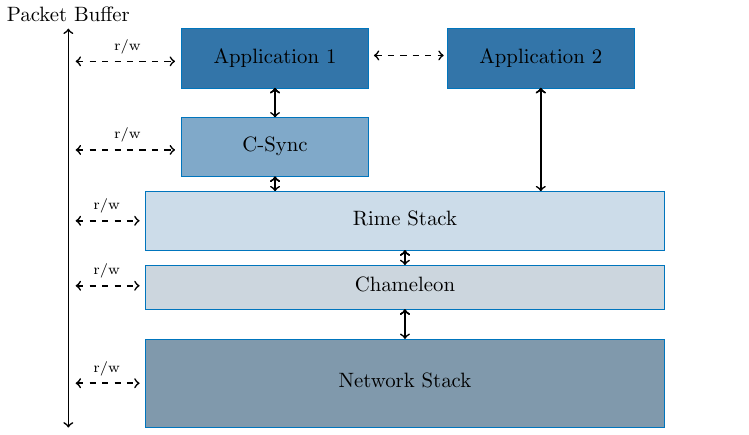}
	\caption{Software architecture of Contiki OS integrated with C-sync.}
	\label{fig:contiki}
\end{figure}

Rime stack is a versatile layer that provides various implementations of communication primitives such as periodic broadcasts, etc. 
This layer handles the packet buffer queuing in case the channel is busy due to CSMA-CA.
Rime also offers a variant of the broadcast primitive called polite announcement that reduces the messages within the radio range of a node by monitoring the channel for any duplicate messages.
If a duplicate message is received, the scheduled transmission packet is dropped; and the transmission continues once the channel is free or if a different message is received.

C-sync runs concurrently over the Rime stack. 
Applications that require energy-efficient and resilient time synchronization for their application may use C-sync, while other applications can directly communicate using the Rime stack.
During the idle phase of C-sync, the applications directly take over the network stack till the next scheduled consensus synchronization/discovery state using interrupts and function callbacks. 
It is important to note that all the layers operate on the same packet buffer and it is passed among the layers every time it is populated or depopulated depending on the direction of the packet. 
For comparison purposes, we have also implemented the Gradient Time Synchronization Protocol (GTSP)~\cite{5211944} on Contiki.
GTSP was chosen for comparison as the representative protocol among the class of decentralized solutions since it forms the basis of averaging and consensus features used in the other recent solutions.
GTSP is also the only decentralized synchronization protocol that has shown feasibility for hardware testing for generic wireless sensor networks.

\paragraph{Hardware clock}
The MSP430 microcontroller consists of a low-frequency 32 kHz crystal oscillator and a high-frequency 8MHz digitally controlled oscillator (DCO).
In our protocol, the hardware clock is derived from a combination of crystal oscillator and the DCO using two 16-bit registers coupled with a clock divider circuit (1/64).
Hence, timer-A extracts the value of the crystal oscillator with a maximum frequency of 512 Hz and a resolution of approximately \SI{30}{\micro\second}, while timer-B extracts the value of DCO operating at a frequency of 524,288 Hz and yields a resolution of approximately \SI{1.9}{\micro\second}.
However, the DCO is very unstable and is susceptible to temperature and voltage changes with a drift of up to 20\%.


\begin{figure}
	\centering
	\centering
	\begin{subfigure}[b]{0.5\textwidth}
		\includegraphics[trim={0 7cm 0 0},clip,width = 0.8\columnwidth]{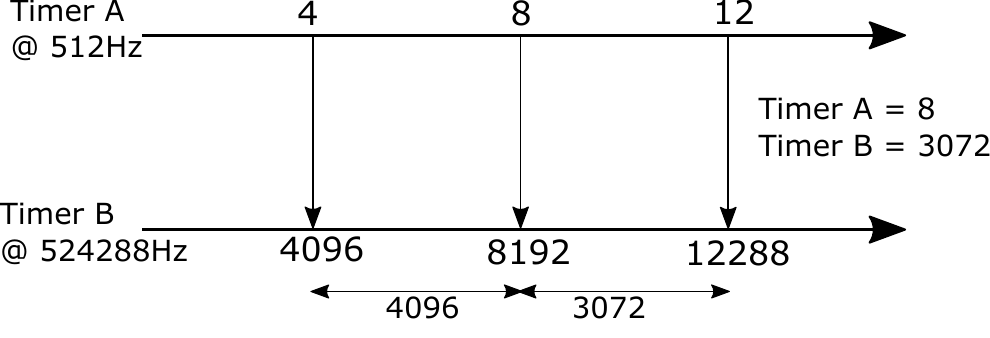}
		\caption{}
		\label{fig:timer_a}
	\end{subfigure}
	
	\begin{subfigure}[b]{0.5\textwidth}
		\includegraphics[trim={0 0 0 7cm},clip,width = 0.8\columnwidth]{figures/timer_compensation.pdf}
		\caption{}
		\label{fig:timer_b}
	\end{subfigure}
	\caption{Compensation of unstable high-frequency DCO using the stable low-frequency crystal oscillator.}
	\label{fig:timer}
\end{figure}

To get a high-resolution stable clock, the crystal oscillator is used to sample the DCO at regular intervals and obtain the DCO drift factor as shown in Figure~\ref{fig:timer}.
Every 4 ticks of timer-A correspond to 4096 ticks of timer-B: this definition is used to compute the correct value of timer-B independent of any variation of the DCO as seen in Figure~\ref{fig:timer_a}. 
Ideal clock and DCO variations cause timer-B variations as shown in Figure~\ref{fig:timer_b}.
The drift of DCO is computed as the ratio of the estimated value $tb_e$ over the actual value of timer-B $tb_a$ as: 
\begin{equation}
\mathrm{drift} = tb_e/tb_a
\end{equation}

This drift factor is multiplied with the current timer-B value to obtain the exact hardware clock value.
With the example shown in Figure~\ref{fig:timer}, the drift factor is $4096/3904 = 1.04918$.
The product of drift with clock value would provide the exact clock value as $1.04918 * 3904 = 4096$.

\subsection{Boot time}
The initial boot time of C-sync to switch from discovery state and achieve loose synchronization is dependent on the density of the network as denser networks take longer to converge due to a larger set of unsynchronized neighbors.
Empirically, in our setup, we found the initial transition to the election phase was roughly 10x the ST interval. Hence, the time to complete the clustering phase from the start of the discovery phase is 10+4 ST intervals.

\subsection{Fault detection and correction}
\label{sec:fault}
The fault detection and correction steps in C-sync are tested by introducing a fault into one of the nodes.
Since the fault handling mechanism is similar across all node types (CH, CB or CM), fault at one type of node is representative of all the other types of nodes. 
With the CB nodes monitoring the flow and correctness of the information, the fault detection at either of three node types follows a similar pattern with exceptions discussed in Section~\ref{sec:exceptions}.
Additionally, the fault handling in a single cluster is representative of the fault handling in the entire network as each cluster counters the fault the same way as any other cluster.

\begin{figure}
	\centering
	\includegraphics[width = 0.9\columnwidth]{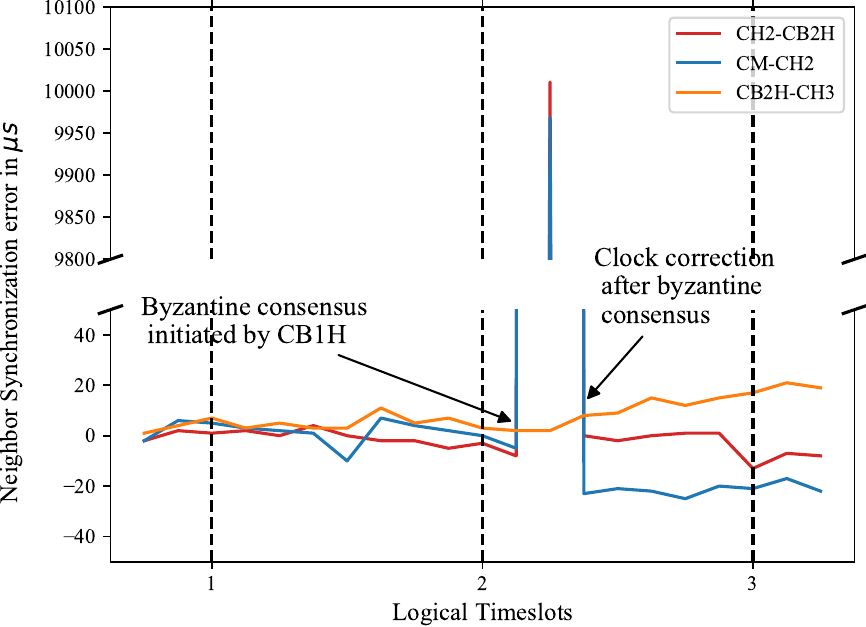}
	\caption{The synchronization error of different nodes within a cluster having a faulty CH. The error remains within one logical time slot without any impact on network synchronization.}
	\label{fig:byzantine}
\end{figure}

We consider an example of three clusters with CH nodes CH1, CH2 and CH3 connected by the bridge node CB1H and CB2H similar to Figure~\ref{general_nw}, where the information flow goes in the order of CH1, CB1H, CH2, CB2H and CH3 from left to right.
Let us suppose CH2 is a faulty node in the network.
Since only nodes of the CH2 cluster are impacted by the fault, we consider only the common nodes in the CH2 cluster denoted as CM.
The threshold for an error to initiate byzantine consensus is set to $500 \mu s$.
Synchronization error is measured as the difference between the clock ticks in the logical clocks of the receiving node and the neighbor node.
CH2 introduces a synchronization error of $10000 \mu s$ in its time slot (slot $2$) after receiving a message from LC, as seen in Figure~\ref{fig:byzantine}.
The CB2H node tries to synchronize to this initially resulting in a high error.
The high error is shown with a broken axis on the plot where both CB and CM are impacted by faulty time information.
The vertical lines on the plot for CH2-CB2H and CM-CH2 converge to the point in the upper half of the plot.
However, byzantine consensus messages with the correct clock value are flooded within the cluster by the CB1H node of the first cluster as shown in Figure~\ref{fig:byzantine}.
Both the CM and CBH nodes reset their clocks to the information in the byzantine consensus messages within the same slot $2$. 
The logical timeslot represents the duration of a synchronization slot, \emph{i.e.,} a combination of consensus synchronization state and the idle state.

CB2H transmits the updated correct clock information to CH3 in the next time slot.
As seen from Figure~\ref{fig:byzantine}, the neighboring cluster head CH3 is not impacted by the fault in the cluster with CH2.
Hence, the fault is contained within a single cluster and recovers with byzantine consensus as expected from the discussion in Section~\ref{sec:fault_handling}.
It is important to note that a fault (if any) in the clustering phase, gets detected in the consensus phase as the error and fault tolerance are a consequence of nodes elected in the clustering phase.

\subsection{Energy efficiency for Fault-tolerance}
\label{sec:energy}
To demonstrate the efficient operation of C-sync, we conducted experiments to measure the synchronization error and the power consumed for achieving neighbor synchronization as shown in Figure~\ref{fig:powervserror}.
Neighbor nodes are a part of the clustered architecture, including, but not limited to communication between CM-CH, CB-CH and CH-CB-CH.
A scatter plot of average neighbor synchronization error in $\mu s$ on the y-axis against average power consumption of each node in $mW$ on the x-axis after the protocols move to the idle phase is plotted.
Each point on the plot represents an averaged value of power and synchronization error for each node in the network over the entire duration of the experiment.
Measured values of synchronization error are plotted without taking an absolute value, \emph{i.e.,} the offset can be both positive or negative.

Although the error and the power are not correlated, they provide an intuition on the synchronization protocol performance to achieve low synchronization error and power consumption.
Similar to C-sync, the GTSP idle phase begins at the end of the discovery phase since GTSP does not employ a state machine in its protocol.


\begin{figure}
	\centering
	\includegraphics[width=0.95\linewidth]{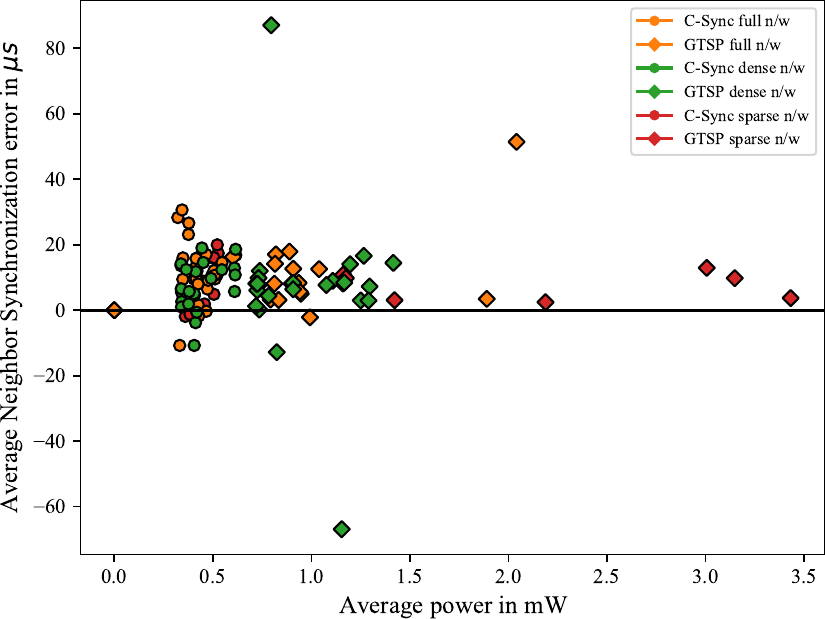}
	\caption{Scatter plot of average neighbor synchronization error and the average power consumption for each node over different topologies that are densely distributed, sparsely distributed and using the entire testbed compared for C-sync and GTSP.}
	\label{fig:powervserror}
\end{figure}

The scatter plot consists of tests conducted in three different topologies formed by utilizing different configurations of the testbed.
The summary of the results of average neighbor synchronization error (offset) and average network power consumption (power) with corresponding standard deviations is shown in Table~\ref{tab:results}.

The dense distribution of nodes is formed by using the different pockets of nodes closely grouped together on each floor of the building (denoted as dense n/w in the figure).
GTSP performs best in a dense environment due to the closeness of the nodes and, hence, converges quickly with a synchronization error of $7.05 \mu s$. 
However, due to the continuous exchange of messages, there is a high power consumption of $0.98 mW$. 
C-sync forms larger clusters with more CM nodes in a denser environment.
Hence, more nodes remain in sleep mode while achieving a similar synchronization error of $7.52 \mu s$ and low power consumption of $0.43 mW$, having a reduction of power by 56.12\% compared to GTSP.

The sparse distribution of nodes is formed by utilizing a few nodes from each floor to communicate with each other (denoted as sparse n/w in the figure).
Due to the sparse distribution, the convergence speed of the 1-hop synchronization algorithm is impacted while multi-hop synchronization performs better~\cite{PANIGRAHI2017124}.
Hence, due to the neighbor-only synchronization, GTSP incurs a longer convergence time with a synchronization error of $7.67 \mu s$ and power consumption of $1.98 mW$.
However, C-sync synchronizes to the closest neighbors to form smaller clusters while letting farther nodes become CB nodes to other similar clusters.
Hence, there is a multi-hop synchronization across multiple clusters yielding a synchronization error of $10.05 \mu s$ and power consumption of $0.48 mW$ which is 75.75\% lower than GTSP.
This way, there is a faster convergence, yielding significant power reduction.

\begin{table}[]
	\caption{Results of average neighbor synchronization error and power consumption measured (with associated standard deviation) over 30 minutes across different topologies.}
	\label{tab:results}
	\resizebox{\columnwidth}{!}{%
	\begin{tabular}{|l|l|l|l|}
		\hline
		\centering
		& Dense Network & Sparse Network & Full Network \\ \hline
		GTSP sync. error ($\mu s$)   & 7.05 (24.51)        & 7.67 (3.73)         & 10.13  (11.64)      \\ \hline
		C-Sync sync. error ($\mu s$) & 7.52 (7.22)       & 10.05 (6.78)       & 12.24 (9.35)      \\ \hline
		GTSP Power ($mW$)   & 0.98 (0.24)       & 1.98 (0.97)        & 0.89 (0.48)      \\ \hline
		C-Sync Power ($mW$)  & 0.43 (0.1)       & 0.48 (0.07)        & 0.43 (0.06)     \\ \hline
	\end{tabular}
	}
\end{table}

Combining both the above environments, utilizing all the nodes of the testbed provides us the results for a full network (denoted as full n/w in the figure).
GTSP synchronizes to the mixed environment yielding an average synchronization error of $10.13 \mu s$ and consuming $0.89 mW$.
C-sync forms multiple clusters with a mixture of large and small clusters connected via cluster bridges.
C-sync achieves an average power consumption of $0.43 mW$ which is roughly 51.68\% of the power consumption of GTSP in the full network topology.
The synchronization error is similar to GTSP with C-sync having an average synchronization error of $10.13 \mu s$.
Both protocols take a longer time than the dense and sparse configurations to achieve synchronization due to the higher number of nodes in the network.

The difference in average synchronization error between C-sync and GTSP across various topologies is at a maximum of about $2.4 \mu s$. This is equivalent to one tick of the clock at a resolution of $1.9 \mu s$.
This difference can be attributed to the measurement error and hence, can be concluded that the accuracy of C-sync is similar to that of GTSP.

C-sync consistently achieves the roughly same deviation in synchronization error, whereas GTSP experiences a larger deviation depending on the topology.
The fixed set of states and the leader-follower synchronization in C-sync ensures that the protocol converges and moves into the idle state faster, leading to lower variation in power consumption.

\subsection{Synchronization Error to Local Center}
\label{sec:error}
As described in the protocol, C-sync is able to restrain the distance of any node to the time source within a pre-defined parametric threshold.
In order to demonstrate this experimentally, we constructed a chain topology to emulate a multi-hop network as shown in Figure~\ref{fig:chain_topo}.
Ignoring all the CM nodes except for the ones on the corner CH nodes, there are 13 nodes in the chain. 

\begin{figure}
	\centering
	\includegraphics[width = 0.9\columnwidth]{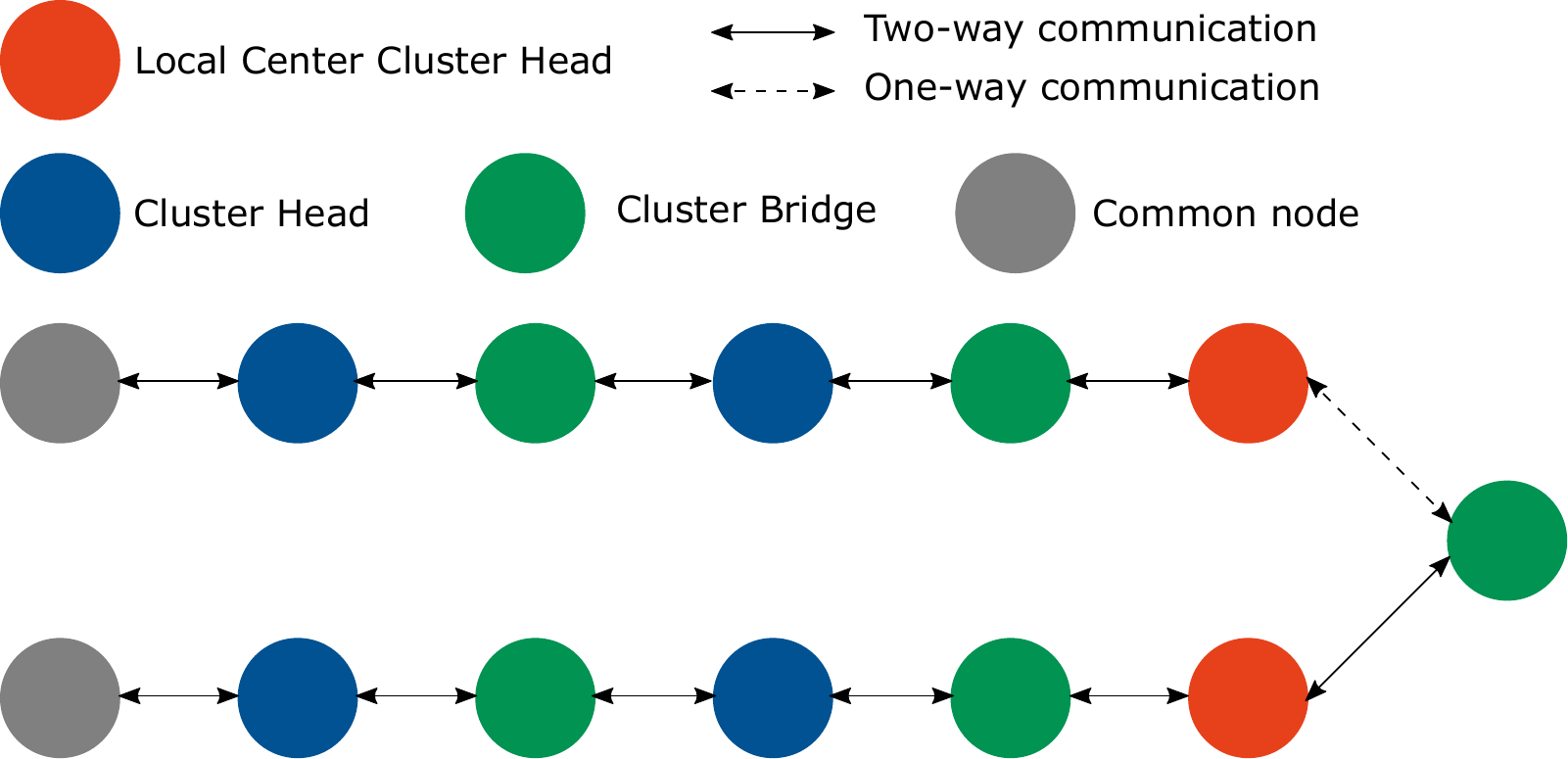}
	\caption{Chain of Cluster Heads and Cluster Bridge ignoring the Common nodes to demonstrate a multi-hop network.}
	\label{fig:chain_topo}
\end{figure}

Letting C-sync run on every node, a chain of CH and CB was established, with each node (except the CM nodes) connected to two neighbors.
The synchronization error from the LCs to each of the nodes associated with the LCs is plotted over the number of hops it receives the time information.
As seen in the example, LC is reached from both ends and the edge nodes synchronize to the time information sent by their associated LC over multiple hops.

The CB node connecting both LC receives information from both but synchronizes to only one of them based on the address of the LC (since both LC have the same hop count).
This can be configured to average the time information and further transmitting it for synchronization among LCs when LCs have different slots of transmission.
Since both LCs have the same hop count in our example, both LCs are active in the same time slot. 
Hence, the CB information is not received by the LC in the next time slot.
The entire network gets re-synchronized at the next discovery phase.

For our example network, the maximum number of hops from any node to its LC was observed to be five.
Each slot in the synchronization phase is $300 ms$ and the idle phase is configured to be ten such slots.
Computing the periodicity of messages received where every node is active only for 1 slot, each node receives a synchronization message once every $10 + (5-1)$ slots, \emph{i.e.,} a periodicity of 14 slots.
The ideal minimum synchronization error to the LC in the case of C-sync is the duration of 1 instruction (1 tick) of recording MAC-layer timestamping.
Based on the resolution of the clock, it is approximately 1.9$\mu s$.
Using the information of maximum hop count and the minimum synchronization error, the worst-case calculated synchronization error to the LC can be calculated from our previous definition of $\eta \cdot \tau \cdot \delta$ as $(1.9 \cdot 14 \cdot 5) = 133\mu s$.
However, this limit reduces if the number of hops of a node is lower than five.
The experimental measurement for the synchronization error to the LC over multiple hops is plotted against the estimated value as shown in Figure~\ref{fig:bounded_error}.

 \begin{figure}
	\centering
	\includegraphics[width = 0.9\linewidth]{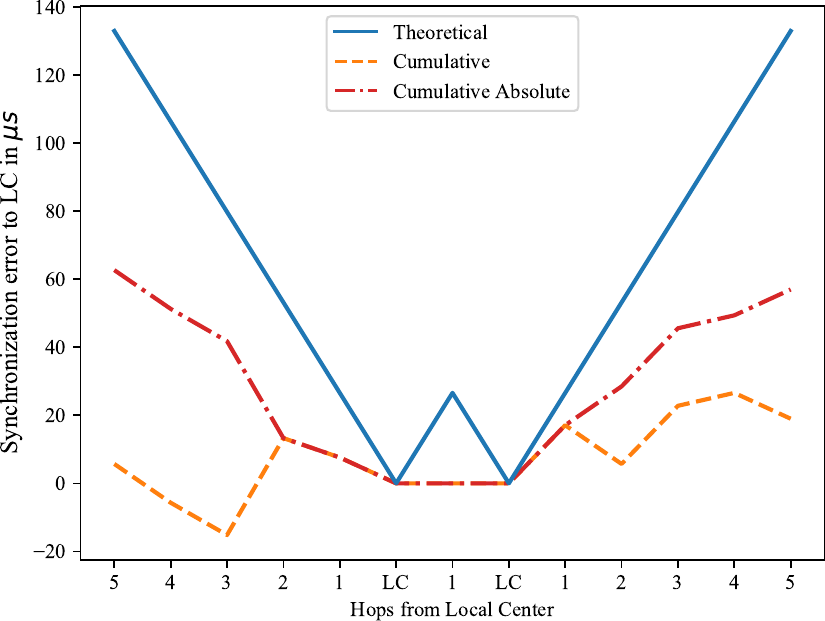}
	\caption{Synchronization error is bounded by the established maximum number of hops from the local center.}
	\label{fig:bounded_error}
\end{figure}

At the LC, the error is zero since its own clock is the reference clock.
As seen from the plot, both the cumulative error and the absolute cumulative error (converting negative offset to positive) were measured and found to stay below the computed theoretical value.

\section{Related Work}\label{literature}


The existing literature can be classified based on the availability of fault handling mechanisms in the protocols.
Protocols with dedicated fault-tolerant features are described under secure synchronization solutions while the rest of the protocols are grouped together as generic synchronization solutions.

\subsubsection*{Secure synchronization solutions}
Ganeriwal \emph{et. al.} investigated the problem of secure time synchronization using shared key encryption and delay threshold-based detection~\cite{10.1145/1380564.1380571}, \emph{i.e.,} the delay between the estimated and the actual time.
Li and Rus counter byzantine faults by adding a layer of cryptographic encoding and decoding during message exchanges but suffer from long convergence times~\cite{1566581}.
Blockchain-based protocols~\cite{8482304,9191420} achieve resilience against byzantine faults at the cost of low accuracy (in the order of seconds) and high computation overhead.
Max-consensus was used to estimate the clock difference to a threshold value to detect byzantine nodes but achieves a low synchronization accuracy~\cite{6520853}.
Sundial~\cite{258951} was proposed for fault-resilient synchronization between data centers by combining a hardware-based detection and software-based reconfiguration for hardware to handle faults.
However, Sundial requires additional hardware (in contrast to resource-constrained IoT nodes) for fault detection while taking substantial time for reconfiguration change transmission from the reference node to an actual change in hardware.
Temporal correlation of messages was used between neighboring nodes to correct synchronization errors~\cite{4509612}.
However, error detection and correction require the exchange of a significant number of messages resulting in communication overhead.
While assuming a trusted resource-rich reference node, digital signatures and message filters were used as validation tools~\cite{5345607}. Secure synchronization protocols focus on ensuring resilience to faults but do not cater to accuracy and energy efficiency due to complex fault handling mechanisms. Additionally, most of the works on secure synchronization solutions are centralized with the assumption that the reference node cannot be faulty. Furthermore, they have not been hardware-proven.

\subsubsection*{Generic Synchronization solutions} 
Historically, the Global Positioning System (GPS) or Network-Time Protocol (NTP)~\cite{103043} has been used for time synchronization in networks. 
However, these protocols are not applicable for resource-constrained nodes.
Precision Time Protocol (PTP) uses a master-slave architecture to synchronize the clock rate at the network layer~\cite{9120376}. PTP requires custom PTP-compatible hardware modules for synchronization and is susceptible to de-synchronization from packet delays.
Reference Broadcast Synchronization (RBS)~\cite{Elson} synchronized a set of receivers to minimize sender-side uncertainties while Time-sync Protocol for Sensor Networks (TPSN) used peer-to-peer synchronization with MAC-layer timestamps for both sender and receiver nodes~\cite{Ganeriwal:2003:TPS:958491.958508}.
However, both protocols cannot handle ad-hoc networks and have no compensation for clock drifts. Flooding-based protocols~\cite{Maroti,6777583,5779066} utilize a central reference node that periodically transmits a large number of synchronization messages in quick succession to which all nodes of the network synchronize at different hops. 
Pulsesync protocol~\cite{6777583} uses rapid flooding to reduce the flooding latency and improves upon FTSP~\cite{Maroti}.
Meanwhile, Glossy protocol achieves synchronization using constructive interference of modulated signals with a temporal displacement within a threshold~\cite{5779066}. This necessitates nodes to be equipped with high-quality radios with low noise and distance between nodes to achieve the synchronization threshold.
Centralized solutions have a single-point failure when the reference node fails, leading to downtime before a new reference node gets elected.
It is important to note that the inherent reliance on the reference node by all nodes of the network leads to a single-point of failure and constant re-configuration in presence of faults.
The delay due to regular re-configurations could be catastrophic in critical real-time applications such as electric grids, etc.~\cite{7397831}.
Other solutions add to hardware overhead with specialized circuits and timers to reduce the jitter from the existing hardware clocks~\cite{8277300,6728881}.
Decentralized solutions do not rely on a single reference node to achieve time synchronization.
To this end, Gradient Time Synchronization Protocol (GTSP)~\cite{5211944} synchronizes precisely among the neighbors by estimating a global clock formed by an average of drift and offset among 1-hop neighbors.
Based on the network topology and placement of nodes, clustering coupled with consensus has been proposed to synchronize the nodes~\cite{7504162,6926769}.
Wu~\emph{et. al.} use the LEACH~\cite{926982} clustering protocol which assumes a synchronized network for communication while Wang~\emph{et. al.} assume a fixed topology with a fixed state for all nodes without any communication delays.
Both cluster-based protocols are not resilient against faults and cannot adapt to dynamic changes in the network.
Using the concept of consensus from control theory~\cite{1333204}, solutions compute a common offset and drift in a tree structure~\cite{SCHENATO20111878, Xie2018133}.
Although synchronization in the consensus mechanism completes in a fixed period, they are compute-intensive with a growing consensus convergence time as the network scales.

Among the available synchronization solutions, the GTSP protocol has an inherent fault tolerance due to averaging of individual time information, \emph{i.e.,} a single faulty node does not impact the overall average.
Additionally, GTSP is the only decentralized protocol with high accuracy for generic wireless networks that had been proven on hardware.
However, GTSP suffers from high power consumption due to the continuous exchange of messages with every neighbor during synchronization with a wider network impact in presence of faults.
Our proposed solution, C-sync, has a slightly longer initial convergence time due to the overhead of establishing the clustered architecture from a completely decentralized structure as discussed in Section~\ref{clustering}.
Exploiting this network structure, C-sync can achieve significant energy savings by a limited exchange of messages for maintaining synchronization and handling faults.
A fault in C-sync gets isolated to the specific cluster/clusters within which nodes can operate without a reference or get a new reference node without impacting the rest of the network.
Hence, we chose to compare C-sync against GTSP for the synchronization accuracy and energy efficiency on a hardware platform as presented in Section~\ref{sec:energy}.
To the best of our knowledge, there is no time synchronization solution that achieves similar fault resilience with energy efficiency as C-sync.

\section{Conclusion and Future work}\label{conclusions}
In this paper, we presented C-sync, a clustering-based decentralized time synchronization protocol that is both resilient to faults and energy-efficient.
C-sync maintains suitable accuracy and uses the clustered architecture to enable more nodes of the network to remain in sleep mode.
This architecture of C-sync paves way for design of resilient and scalable real-time applications on decentralized networks.
The implementation has been done on Contiki and a hardware testbed.

The fault handling mechanism with byzantine consensus was described and demonstrated experimentally by introducing a fault in a simple network topology that can be scaled.
We illustrated through experiments that C-sync achieves significantly lower power consumption compared to GTSP while attaining similar accuracy.
Additionally, the concept of local centers was introduced and their role in restricting the synchronization error in the network was demonstrated.

The modular nature of the implementation can be used to expand the future applicability of the C-sync to heterogeneous platforms.
Additionally, we aim to include more types of faults to achieve a more comprehensive fault resilience.
\section{Acknowledgements}
This work was financially supported in part by the Singapore National Research Foundation under its Campus for Research Excellence And Technological Enterprise (CREATE) programme.


\bibliographystyle{IEEEtran}
\bibliography{references}

\end{document}